\documentclass[conference,a4paper]{IEEEtran}

\usepackage{amsfonts}
\usepackage{amssymb}
\usepackage{amsmath}
\usepackage{bbm}
\usepackage{dsfont}
\usepackage{graphicx}

\newcommand{\bra}[1]{\langle #1|}
\newcommand{\ket}[1]{|#1\rangle}
\newcommand{\braket}[2]{\langle #1|#2\rangle}
\DeclareMathOperator{\Tr}{Tr}

\newtheorem{theorem}{Theorem}

\newtheorem{remark}{Remark}

\begin{document}

\sloppy

\title{Lov{\'a}sz's Theta Function, R{\'e}nyi's Divergence\\ and the Sphere-Packing Bound}

\author{
  \IEEEauthorblockN{Marco Dalai}
  \IEEEauthorblockA{Department of Information Engineering\\
    University of Brescia - Italy\\
    Email: marco.dalai@ing.unibs.it} 
}



\maketitle

\begin{abstract}
Lov\'asz's bound to the capacity of a graph and the the sphere-packing bound to the probability of error in channel coding are given a unified presentation as information radii of the Csisz\'ar type using the R{\'e}nyi divergence in the classical-quantum setting. This brings together two results in coding theory that are usually considered as being of a very different nature, one being a ``combinatorial'' result and the other being ``probabilistic''. In the context of quantum information theory, this difference disappears.
\end{abstract}

\section{Introduction}
One of the central topics in coding theory is the problem of bounding the probability of error of optimal codes for communication over a given channel. Shannon \cite{shannon-1948} introduced the notion of channel capacity $C$, which  represents the largest rate at which information can be sent through the channel with probability of error that vanishes with increasing block-length. He then also introduced \cite{shannon-1956} the notion of zero-error capacity $C_0$ as the largest rate at which information can be sent with probability of error precisely equal to zero. For rates in the range $C_0<R<C$, the probability of error is known to decrease exponentially in the block-length $n$ as
\begin{equation}
P_e\approx e^{-n E(R)},
\end{equation}
where $E(R)$ is the so called reliability function of the channel. While in the the region of high rates the function $E(R)$ is known exactly, in the low rate region little is known about $P_e$; determining both $E(R)$ and $C_0$ is an unsolved problem and only upper and lower bounds for these quantities are known.

Two of the most important contributions to the study of $E(R)$ and of $C_0$, which came respectively in the '60s and in the '70s, are the sphere-packing bound $E(R)\leq E_{sp}(R)$ \cite{shannon-gallager-berlekamp-1967-1} and Lov{\'a}sz's bound $C_0\leq \vartheta$ \cite{lovasz-1979}. These two bounds are usually considered as being the result of totally unrelated methods. In this paper, we show that this is not the case, and that Lov{\'a}sz's result comes as a special case of the sphere-packing bound once we move to the more general context of classical-quantum channels. In order to do that, we extend to the classical-quantum case a result of Csisz\'ar that allows us to express the sphere-packing exponent \cite{dalai-ISIT-2012} in terms of an information radius using the R\'enyi divergence. Lov{\'a}sz's result then emerges naturally as a special case.
This leads to a unified view of two of the most important bounds to $E(R)$ and to $C_0$, 
showing that quantum information theory is a useful tool to attack problems at the intersection of probability and combinatorics in classical information theory.

\section{Classical Channels}

\subsection{Basic notations and definitions}
Let $W(x|y)$, $x\in \mathcal{X}$, $y\in\mathcal{Y}$, be the transition probabilities of a discrete memoryless channel $W :\mathcal{X}\to\mathcal{Y}$, where $\mathcal{X}$ and $\mathcal{Y}$ are finite sets. For a sequence $\mathbf{x}=(x_1,x_2,\ldots,x_n)\in\mathcal{X}^n$ and a sequence $\mathbf{y}=(y_1,y_2,\ldots,y_n)\in\mathcal{Y}^n$, the probability of observing $\mathbf{y}$ at the output of the channel given $\mathbf{x}$ at the input is 
\begin{equation}
W^{(n)}(\mathbf{y}|\mathbf{x})=\prod_{i=1}^n W(y_i|x_i).
\end{equation}
A block code with $M$ messages and block-length $n$ is a mapping from a set $\{1,2,\ldots,M\}$ of $M$ messages onto a set $\{\mathbf{x}_1, \mathbf{x}_2, \ldots, \mathbf{x}_M\}$ of $M$ sequences in $\mathcal{X}^n$. The rate $R$ of the code is defined as $R=\log M/n$.
A decoder is a mapping from $\mathcal{Y}^n$ into the set of possible messages $\{1,2,\ldots,M\}$. If message $m$ is to be sent, the encoder transmits  the codeword $\mathbf{x}_m$ through the channel. An output sequence $\mathbf{y}$ is received by the decoder, which maps it to a message $\hat{m}$. An error occurs if $\hat{m}\neq m$.

Let $Y_m\subseteq \mathcal{Y}^n$ be the set of output sequences that are mapped into message $m$. When message $m$ is sent, the probability of error is
\begin{equation}
P_{e|m}=\sum_{\mathbf{y}\notin Y_m} W^{(n)}(\mathbf{y}|\mathbf{x}_m).
\end{equation}
The maximum error probability of the code is defined as the largest $P_{e|m}$, that is,
\begin{equation}
P_{e,\max}=\max_{m}P_{e|m}.
\end{equation}

Let $P_{e,\max}^{(n)}(R)$ be the smallest maximum error probability among all codes of length $n$ and rate at least $R$.
Shannon's theorem \cite{shannon-1948} states that sequences of codes exists such that $P_{e,\max}^{(n)}(R)\to 0$ as $n\to\infty$ for all rates smaller than a constant $C$, called \emph{channel capacity}, which is given by the expression
\begin{equation}
C=\max_{P}\sum_{x,y}P(x)W(y|x)\log\frac{W(y|x)}{\sum_{x'} P(x') W(y|x')},
\end{equation} 
where the maximum is over all probability distributions on the input alphabet. 

For $R<C$, Shannon's theorem only asserts that $P_{e,\max}^{(n)}(R)\to 0$ as $n\to\infty$.  
For a range of rates  $C_0\leq R \leq C$, the optimal probability of error $P_{e,\max}^{(n)}(R)$ is known to have an exponential decrease in $n$, and it is thus useful to define the \emph{reliability function} of the channel as
\begin{equation}
E(R)=\limsup_{n\to\infty} -\frac{1}{n}\log P_{e,\max}^{(n)}(R).
\label{eq:E(R)_def_class}
\end{equation} 
The value $C_0$ is the so called \emph{zero-error capacity}, also introduced by Shannon \cite{shannon-1956}, which is defined as the highest rate at which communication is possible with probability of error precisely equal to zero. More formally,
\begin{equation}
C_0=\sup\{R \, :\,   P_{e,\max}^{(n)}(R)=0 \mbox{ for some } n\}.
\end{equation}
For $R<C_0$, we may define the reliability function $E(R)$ as being infinite.
Determining the reliability function $E(R)$ (at low positive rates) and the zero-error capacity $C_0$ of a general channel is still an unsolved problem.

\subsection{Reliability and zero-error capacity}
In order to study the zero-error capacity of a channel, it is important to consider when two input symbols or two input sequences are confusable and when they are not.
Note that two input symbols $x$ and $x'$ cannot be confused at the output if and only if the associated conditional distribution $W(\cdot|x)$ and $W(\cdot|x')$ have disjoint supports. Furthermore, two sequences $\mathbf{x}=(x_1,\ldots,x_n)$ and  $\mathbf{x'}=(x_1',\ldots,x_n')$ cannot be confused if and only if there exists at least one index $i$ such that symbols $x_i$ and $x_i'$ are not confusable. For a given channel $W$, it is then useful to define a \emph{confusability graph} $G(W)$ whose vertices are the elements of $\mathcal{X}$ and whose edges are the elements $(x,x')\in\mathcal{X}^2$ such that $x$ and $x'$ are confusable. It is then easily seen that $C_0$ only depends on $G(W)$. Furthermore, for any $G$, we can always find a channel $W$ such that $G(W)=G$. Thus, we may equivalently speak of the zero-error capacity of a channel $W$ or of the \emph{capacity $C(G)$ of the graph} $G$ if $G=G(W)$, and we will use those two notions interchangeably through the paper.
A first upper bound to $C_0$ was obtained by Shannon \cite{shannon-1956}, who upper bounded $C_0$ with the zero-error capacity $C_{FB}$ when perfect feedback is available. He could prove by means of a combinatorial argument that, if $C_0>0$, then 
\begin{equation}
C_{FB} = \max_{P} -\log \max_y \sum_{x: W(y|x)>0} P(x).
\label{eq:C_FB}
\end{equation}
Given a graph $G$, then, the best bound to $C(G)$ is obtained by using the channel $W'$ with $G(W')=G$ which minimizes $C_{FB}$. 
Interestingly enough, this bound can also be obtained by a rather different method that relies on bounding the reliability function $E(R)$. In particular, the so called sphere-packing bound, first derived in \cite{fano-book}  and later rigorously proved in \cite{shannon-gallager-berlekamp-1967-1}, states that $E(R)\leq E_{sp}(R)$, where $E_{sp}(R)$ is defined by 
\begin{align*}
E_{sp}(R) & \geq \sup_{\rho \geq 0 }\left[E_0(\rho)-\rho R \right]\\
E_0(\rho) & = \max_{P} E_0(\rho, P)\\
E_0(\rho,P) & = -\log \sum_{y}\left(\sum_{x} P(x) W(y|x)^{1/{(1+\rho)}}\right)^{1+\rho}.
\end{align*}

The function $E_{sp}(R)$ is finite for all rates $R$ larger than the quantity
\begin{equation}
R_\infty = \max_{P} -\log \max_y \sum_{x:\, W(y|x)>0} P(x),
\label{eq:Rinfty_class}
\end{equation}
which implies that $E(R)$ is finite for $R>R_\infty$ and thus that $C_0\leq R_\infty$. Interestingly enough, we see that if $C_0>0$ then $R_\infty=C_{FB}$.
This implies that in all cases of practical interest, Shannon's bound to $C_0$, which was first derived by means of a combinatorial method, can also be deduced from the sphere-packing bound, which is instead derived in a probabilistic setting.

A major breakthrough came with Lov\'asz's 1979 work \cite{lovasz-1979}. 
Given a confusability graph $G$, Lov{\'a}sz calls an \emph{orthonormal representation} of $G$ any set $\{u_x\}_{x\in\mathcal{X}}$ of unit norm vectors in any Hilbert space such that $u_x$ and $u_{x'}$ are orthogonal if symbols $x$ and $x'$ are not confusable. We will use here the \emph{bra-ket} notation $\braket{a}{b}$ for the scalar product between two vectors $a$ and $b$.
He then defines the \emph{value} of a representation $\{u_x\}$ as\footnote{We use a logarithmic version of the theta function so as to make its comparison with rates more straightforward.} 
\begin{equation}
V(\{u_x\})=\min_c \max_x \log \frac{1}{|\braket{u_x}{c}|^2},
\end{equation}
where the minimum is over all unit norm vectors $c$. The vector $c$ that achieves the minimum above is called the \emph{handle} of the representation. Lov\'asz shows that any orthonormal representation satisfies $V(\{u_x\})\geq C_0$. Optimizing over all representations, he thus gives a bound for $C_0$ in the form $C_0\leq \vartheta$, where
\begin{align*}
\vartheta& = \min_{\{u_x\}}\min_c \max_x \log \frac{1}{|\braket{u_x}{c}|^2}
\end{align*}
is the so called Lov\'asz theta function. This result is usually considered to be of a purely combinatorial nature and no probabilistic interpretation seems to have emerged up to now.  It is interesting to note, however, that
 a possible representation for the confusability graph of a channel $W$ can simply be constructed by taking the set of $|\mathcal{Y}|$-dimensional real valued vectors $\{\varphi_x\}$ with components $\varphi_x(y)=\sqrt{W(y|x)}$. As we will show later, the value of this representation $V(\{\varphi_x\})$ is precisely the cut-off rate of the channel, which is never smaller than $C_0$. Clearly, using different channels $W'$ (with $G(W')=G(W)$), we may upper bound $C_0$ with the lowest of their cut-off rates. Nicely enough, it turns out that this would lead precisely to the same upper bound obtained by means of $C_{FB}$ (or $R_\infty$). Lov{\'a}sz's theta function achieves a smaller upper bound to $C_0$ due to the fact that it allows the components of the vectors of a representation to take on negative values. Lov\'asz's approach seems thus to suggest bounding the zero-error capacity by considering the use of quantum-theoretic wave functions in place of classical probability distributions.

\subsection{R\'enyi's Information Radii}

It is known \cite{csiszar-korner-book} that the capacity of a classical channel can be written as an information radius according to the expression
\begin{equation}
C=\min_Q \max_x D(W(\cdot|x)||Q),
\label{eq:minmaxC}
\end{equation}
where $D(\cdot||\cdot)$ is the Kullback-Leibler divergence. This min-max formulation was extended by Csisz{\'a}r \cite{csiszar-1995} to describe the reliability function in the high rate region. Here, since we are only interested in upper bounds to $E(R)$, it is useful to consider the sphere-packing exponent $E_{sp}(R)$, for which Csisz{\'a}r's min-max expression holds with full generality. The function $E_{sp}(R)$ equals the upper envelope of all the lines $E_0(\rho)-\rho R$, and an important quantity is the value $R_\rho=E_0(\rho)/\rho$ at which each of these lines meets the $R$ axis.\footnote{Here, since we also consider the true zero-error capacity $C_0$, we do not adopt Csisz{\'a}r's notation of \emph{channel capacity of order $\alpha$}.}
Given two distributions $Q_1$ and $Q_2$ on the channel output $\mathcal{Y}$, define the \emph{R\'enyi divergence} of order $\alpha\in(0,1)$ of $Q_1$ from $Q_2$ as
\begin{equation}
D_\alpha(Q_1||Q_2)=\frac{1}{\alpha-1}\log \sum_{y} Q_1(y)^\alpha Q_2(y)^{1-\alpha}.
\end{equation}
It is then shown in \cite[Prop. 1]{csiszar-1995} that
\begin{equation}R_\rho=\min_{Q}\max_x D_\alpha(W(\cdot|x)||Q), \quad \alpha=1/(1+\rho).
\label{eq:minmaxRrho}
\end{equation}
Using the known properties of the R\'enyi divergence (see \cite{csiszar-1995}), we find that when $\rho\to 0$ the above expression (with $\alpha\to 1$) gives the already mentioned expression for the capacity \eqref{eq:minmaxC}, while for $\rho\to\infty$ we obtain 
\begin{equation}
R_\infty=\min_{Q}\max_x -\log \sum_{y:W(y|x)>0}Q(y),
\end{equation}
which is the dual formulation of \eqref{eq:Rinfty_class}.

It is evident that there is an interesting similarity between the min-max expression for $R_\rho$ of a channel $W$ and the value of a representation in Lov{\'a}sz' sense. In the next sections, we will show that this similarity is not a simple coincidence. Lov{\'a}sz' bound to $C_0$ and the sphere-packing bound to $E(R)$ are based on the very same idea and can be described in a unified way in probabilistic terms in the context  of quantum information theory. By considering the extension of the sphere-packing bound to classical-quantum channels, we will show that Lov{\'a}sz' bound emerges naturally, in that case, as a consequence of the bound $C_0\leq R_\infty$.

\begin{remark}

A very nice fact, apparently not reported in the literature, is that the usual cut-off rate of a classical channel $W$, evaluated according to equation \eqref{eq:minmaxRrho} with $\alpha=1/2$, is precisely the value $V(\{\varphi_x\})$ of the representation $\{\varphi_x\}$ with $\varphi_x=\sqrt{W(\cdot|x)}$. In this paper, however, we will interpret Lov{\'a}sz's value of a representation $\{u_x\}$ in relation to the rate $R_\infty$ of a pure-state classical-quantum channel with state vectors $\ket{u_x}$. It turns out \cite{dalai-QSP-2012} that the cut-off rate of a classical channel $W$ precisely equals the rate $R_\infty$ of a pure-state classical-quantum channel with state vectors $\ket{\varphi_x}$ as defined above, but the true reason for this equivalence is not yet clear.
\end{remark}

\section{Classical-Quantum Channels}
\label{sec:Basic-quantum}
\subsection{Basic notions and the sphere-packing bound}
We introduce here the minimal notions and results on classical-quantum channels so as to make this paper as self-contained as possible. The interested reader may refer to \cite{hayashi-book-2006} \cite{wilde-2012} for more details.

Following \cite{holevo-2000}, consider a classical-quantum channel with a finite input alphabet $\mathcal{X}$ with associated density operators $S_x$, $x\in\mathcal{X}$ in a finite dimensional Hilbert space\footnote{The $S_x$ can thus be represented as positive semi-definite Hermitian matrices with unit trace.} $\mathcal{H}$. The $n$-fold product channel acts in the tensor product space $\mathcal{H}^{\otimes n}$ of $n$ copies of $\mathcal{H}$. To a codeword $\mathbf{x}=(x_1,x_2,\ldots,x_n)$ is associated the signal state $\mathbf{S}_\mathbf{x}=S_{x_1}\otimes S_{x_2}\cdots\otimes S_{x_n}$.
A block code with $M$ codewords is a mapping from a set of $M$ messages $\{1,\ldots,M\}$ into a set of $M$ codewords  $\mathbf{x}_1,\ldots, \mathbf{x}_M$. The rate of the code is defined as
$R=\frac{\log M }{n}$.

A quantum decision scheme for such a code is a so-called POVM (see for example \cite{wilde-2012}), that is, a  collection of $M$ positive operators\footnote{The operators $\Pi_m$ can thus be represented as positive semi-definite matrices. The notation $\sum_m \Pi_m \leq \mathds{1}$ simply means that $\mathds{1}-\sum_m \Pi_m $ is positive semidefinite. Note that, by construction, all the eigenvalues of each operator $\Pi_m$ must be in the interval $[0,1]$.} $\{\Pi_1,\Pi_2,\ldots,\Pi_M\}$ such that $\sum \Pi_m\leq \mathds{1}$, where $\mathds{1}$ is the identity operator.
The probability that message $m'$ is decoded when message $m$ is transmitted is $P(m'|m)=\Tr \Pi_{m'} \mathbf{S}_{\mathbf{x}_m}$. 
The probability of error after sending message $m$ is
\begin{equation}
P_{e|m}=1-\Tr\left(\Pi_m \mathbf{S}_{\mathbf{x}_m}\right).
\end{equation}
We then define $P_{e,\max}$, $P_{e,\max}^{(n)}(R)$, $C$, $C_0$ and $E(R)$ precisely as in the classical case. 

As in the classical case, we can still express $C_0$ as the capacity of a confusability graph (see \cite{medeiros2006} for more general results) where, in this case, two input symbols are confusable if and only if $\Tr(S_x S_{x'})>0$. 
In fact, if a code with $M$ codewords satisfies $P_{e,\max}=0$, then for each $m\neq m'$ we must have $\Tr(\Pi_m \mathbf{S}_{\mathbf{x}_m})=1$ and $\Tr(\Pi_m  \mathbf{S}_{\mathbf{x}_{m'}})=0$. This is possible if and only if the signals $ \mathbf{S}_{\mathbf{x}_m}$ and $ \mathbf{S}_{\mathbf{x}_{m'}}$ are orthogonal, that is $\Tr(\mathbf{S}_{\mathbf{x}_m}\mathbf{S}_{\mathbf{x}_{m'}})=0$. But, using the property that $\Tr((A\otimes B)(C\otimes D))=\Tr(AC)\Tr(BD)$, we have
\begin{align}
\Tr(\mathbf{S}_{\mathbf{x}_m}\mathbf{S}_{\mathbf{x}_{m'}}) & =\prod_{i=1}^n\Tr(S_{x_{m,i}}S_{x_{m',i}}).
\end{align}
This implies that $\Tr(S_{x_{m,i}}S_{x_{m',i}})=0$ for at least one value of $i$.
Thus, evaluating the zero-error capacity in the classical-quantum setting amounts to evaluating the capacity of a graph as defined in the previous section. In this sense, there is no difference between classical and classical-quantum channels and, given a graph $G$, we can interpret the capacity $C(G)$ as either the zero error capacity $C_0$ of a classical or of a classical-quantum channel with that confusability graph. (For recent results on the zero-error communication via general quantum channels see \cite{duan-severini-winter-2013} and references therein).

For classical-quantum channels, bounds to the reliability function $E(R)$ have been developed which partially match those of the classical case. 
Lower bounds to the reliability function were obtained in \cite{burnashev-holevo-1998} and \cite{holevo-2000}, while upper bounds have remained relatively unexplored until recently. For general $R>0$, the first upper bound to $E(R)$ was obtained in \cite{dalai-ISIT-2012} as an extension of the classical sphere-packing bound of \cite{shannon-gallager-berlekamp-1967-1}. The bound can be stated as follows.

\begin{theorem}[Sphere Packing Bound \cite{dalai-ISIT-2012}\cite{dalai-QSP-2012}]
For all positive rates $R$ and all positive $\varepsilon < R$,
\begin{equation}
E(R)\leq E_{sp}(R-\varepsilon),
\end{equation}
where $E_{sp}(R)$ is defined by the relations
\begin{eqnarray}
E_{sp}(R) & = & \sup_{\rho \geq 0} \left[ E_0(\rho) - \rho R\right]\label{eq:Esp}\\
E_0(\rho) & = & \max_{P}E_0(\rho,P)\label{eq:E0rho}\\
E_0(\rho,P) & =& -\log\Tr\left( \sum_{x} P(x) S_x^{1/(1+\rho)} \right)^{1+\rho}.
\label{eq:E0rhoq}
\end{eqnarray}
\end{theorem}

\subsection{Quantum R\'enyi's Information Radii}

We now extend Csisz{\'a}r's result to give a characterization of the sphere packing bound for classical-quantum channels in terms of R\'enyi's information measures. Given two density operators $F_1$ and $F_2$ in $\mathcal{H}$, and $\alpha\in (0,1)$, define the R\'enyi divergence of order $\alpha$ of $F_1$ from $F_2$ as
\begin{equation}
D_\alpha(F_1||F_2)=\frac{1}{\alpha-1}\log \Tr F_1^\alpha F_2^{1-\alpha}.
\end{equation}
As in the classical case, for $\rho>0$, let then
\begin{equation}
R_\rho=E_0(\rho)/\rho.
\end{equation}
Then we have the following result.
\begin{theorem}
\label{th:Rrho}
For a classical-quantum channel with states $S_x$. $x\in\mathcal{X}$ and $\rho>0$, the rate $R_\rho$ defined above satisfies
\begin{equation}
R_\rho=\min_{F}\max_{x} D_\alpha(S_x||F), \quad  \alpha=1/(1+\rho).
\end{equation} 
\end{theorem}
\begin{proof}
Setting $\alpha=1/(1+\rho)$, we can write
\begin{equation}
R_\rho =  \max_P \frac{1}{\alpha-1}\log\left[ \Tr\left( \sum_x P(x) S_x^{\alpha} \right)^{1/\alpha}\right]^{\alpha}
\end{equation}
and, defining $A(\alpha,P) =\sum_x P(x) S_x^{\alpha}$, we can write
\begin{equation}
R_\rho =  \max_P \frac{1}{\alpha-1}\log \| A(\alpha,P) \|_{1/\alpha},
\label{eq:rrhoins}
\end{equation}
where $\|\cdot\|_{r}$ is the Schatten $r$-norm. From the H\"older inequality we know that, for any  positive operators $A$ and $B$, we have
\begin{equation}
\|A\|_{1/\alpha}\|B\|_{1/(1-\alpha)}\geq \Tr(AB)
\end{equation}
with equality if an only if $B=\gamma A^{1-1/\alpha}$ for some scalar coefficient $\gamma$. Thus we can write
\begin{equation}
\|A\|_{1/\alpha}=\max_{\|B\|_{1/(1-\alpha)}\leq 1} \Tr(AB),
\end{equation}
where $B$ runs over positive operators in the unit ball in the $(1/(1-\alpha))$-norm.
Using this expression for the Schatten norm we obtain
\begin{align}
R_\rho & =  \max_P  \frac{1}{\alpha-1}\log \max_{\|B\|_{1/(1-\alpha)}\leq 1} \Tr( A(\alpha,P) B)\\
& =\frac{1}{\alpha-1}\log  \min_P \max_{\|B\|_{1/(1-\alpha)}\leq 1} \Tr\left( \sum_x P(x) S_x^{\alpha} B \right).
\end{align}
In the last expression, the minimum and the maximum are both taken over convex sets and the objective function is linear both in $P$ and $B$. Thus, we can interchange the order of maximization and minimization to get
\begin{align}
R_\rho & = \frac{1}{\alpha-1}\log \max_{\|B\|_{1/(1-\alpha)}\leq 1}\min_P  \sum_x P(x) \Tr\left(S_x^{\alpha} B \right) \\
&= \frac{1}{\alpha-1}\log\max_{ \|B\|_{1/(1-\alpha)}\leq 1}\min_x  \Tr\left(S_x^{\alpha} B \right).
\end{align}
Now, we note that the maximum over $B$ can always be achieved by a positive operator, since all the $S_x^\alpha$ are positive operators. Thus, we can change the dummy variable $B$ with $F=B^{1/(1-\alpha)}$, where $F$ is now a positive operator constrained to satisfy $\|F\|_1\leq 1$, that is, it is a density operator. Using $F$, we get
\begin{align}
R_\rho & = \frac{1}{\alpha-1}\log \max_{F}\min_{x}  \Tr\left(S_x^{\alpha} F^{1-\alpha} \right) \\
& =  \min_{F}\max_{x} \frac{1}{\alpha-1} \log \Tr\left(S_x^{\alpha} F^{1-\alpha} \right)\\
& =  \min_{F}\max_{x} D_\alpha (S_x||F).
\end{align}
where $F$ now runs over all density operators. 
\end{proof}

It is obvious that, if all operators $S_x$ commute, which means that the channel is classical, than the optimal $F$ is diagonal in the same basis where the $S_x$ are, and we thus recover Csisz\'ar's expression for the classical case. Furthermore, for $\rho\to 0$ (that is, $\alpha\to 1$) we obtain the expression of the capacity as an information radius 
already established for classical-quantum channels \cite{hayashi-nagaoka-2003}. When $\rho=1$ (that is, $\alpha=1/2$) then, we obtain an alternative expression for the so called quantum cut-off rate \cite{ban-kurokawa-hirota-1998}. 
The most important case in our context, however, is the case when $\rho\to\infty$ (that is, $\alpha\to 0$).
Taking the limit in Theorem \ref{th:Rrho}, letting $S_x^0$ be the projector in the subspace of $S_x$, we obtain
\begin{equation}
R_\infty=\min_{F}\max_{x} \log\frac{1}{\Tr\left(S_x^0 F \right)},
\label{eq:minmaxQRinfty}
\end{equation}
where the minimum is again over all density operators $F$. Note that the argument of the min-max in \eqref{eq:minmaxQRinfty} coincides with $D_{\min}(S_x|| F)$ according to the definition of $D_{\min}$ introduced in \cite{datta2009}.

The analogy with the Lov{\'a}sz theta function becomes evident if we consider a special case of \eqref{eq:minmaxQRinfty}. 
Assume that the states $S_x$ are pure and set $S_x=\ket{u_x}\bra{u_x}$. Consider for a moment the search for the optimum $F$ when restricted to rank-one operators, that is $F=\ket{f}\bra{f}$. We see that in this case we can write $\Tr(S_x^0F)=|\braket{u_x}{f}|^2$. When searching over all possible $F$, we thus find that for this channel we have
\begin{equation}
R_\infty\leq V(\{u_x\}).
\label{eq:RinftyvsV}
\end{equation}
Hence, we see that Lov{\'a}sz's bound $C_0\leq V(\{u_x\})$ can be deduced as a consequence of $C_0\leq R_\infty$. For a given graph $G$, one may want to bound $C(G)$ with the smallest $R_\infty$ over all channels with confusability graph $G$. This is discussed in the next section section.
%

\section{Sphere Packing and the Lov\'asz Theta Function}

For a given confusability graph $G$, inspired by \eqref{eq:minmaxQRinfty},
we define a \emph{representation} of $G$ any set of projectors $\{U_x\}$ such that $U_x U_{x'}=0$ if symbols $x$ and $x'$ cannot be confused. Furthermore, we introduce an alternative definition of \emph{value}
\begin{equation}
V_{sp}(\{U_x\})=\min_{F}\max_x \log\frac{1}{\Tr\left(U_x F \right)},
\end{equation}
where the minimum is over all density operators $F$. The optimal $F$ will be called again the handle of the representation. We can then finally define the quantity.
\begin{equation}
\vartheta_{sp}=\min_{\{U_x\}}\min_{F}\max_{x} \log\frac{1}{\Tr\left(U_x F \right)},
\end{equation}
where $\{U_x\}$ runs over all representations of the graph $G$. We then have the following result.
\begin{theorem}
\label{th:mytheta}
For any graph, we have
\begin{equation}
C(G)\leq \vartheta_{sp}\leq \vartheta.
\end{equation}
\end{theorem}
\begin{proof}
The fact that $\vartheta_{sp}\leq \vartheta$ is obvious, since Lov\'asz's $\vartheta$ is obtained by restricting the minimization in the definition of $\vartheta_{sp}$ to rank-one projectors $U_x=\ket{u_x}\bra{u_x}$ and handle $F=\ket{f}\bra{f}$. That $C_0\leq \vartheta_{sp}$ should be clear in light of the above discussion on the bound $E(R)\leq E_{sp}(R)$. It is instructive, however, to present a self-contained proof along the same argument used by Lov\'asz.

Consider an optimal representation $\{U_x\}$ and, to a sequence of symbols $\mathbf{x}=(x_1,x_2,\ldots,x_n)$, associate the operator (projector) $\mathbf{U}_\mathbf{x}=U_{x_1}\otimes U_{x_2}\cdots\otimes U_{x_n}$.
Consider then a zero-error code with $M$ codewords of length $n$, $\mathbf{x}_1,\ldots,\mathbf{x}_M$, and their associated projectors $\mathbf{U}_{\mathbf{x}_1},\ldots,\mathbf{U}_{\mathbf{x}_M}$. 
Then, as proved before, for $m\neq m'$ we have $\Tr(\mathbf{U}_{\mathbf{x}_m}\mathbf{U}_{\mathbf{x}_{m'}})=0$.
Hence, since the states $\{\mathbf{U}_{\mathbf{x}_m}\}$ are orthogonal projectors, we clearly have
\begin{equation}
\sum_{m=1}^{M}\mathbf{U}_{\mathbf{x}_m}\leq \mathds{1},
\label{eq:subpartition}
\end{equation}
where $\mathds{1}$ is the identity operator. Consider now the state $\mathbf{F}=F^{\otimes n}$ where $F$ is the handle of the representation $\{U_x\}$. Note that, for each $m$, we have
\begin{align*}
\Tr (\mathbf{U}_{\mathbf{x}_m} \mathbf{F}) & = \prod_{i=1}^n \Tr (U_{x_{m,i}} F)\\
& \geq e^{-n \vartheta_{sp}}.
\end{align*}
So, using \eqref{eq:subpartition}, we deduce that 
\begin{eqnarray}
1 & \geq &  \sum_{m=1}^M \Tr (\mathbf{U}_{\mathbf{x}_m} \mathbf{F})\\ 
& \geq & Me^{-n \vartheta_{sp}}.
\end{eqnarray}
and hence that $M\leq e^{n\vartheta_{sp}}$.
\end{proof}

\textbf{Note added in the final version:} 
Schrijver \cite{schrijver2013} has observed that Lemma 4 and Corollary 1 in \cite{lovasz-1979}  apply \emph{mutatis mutandis} with our definitions of representation and of $\vartheta_{sp}$. Then, Theorem 5 in \cite{lovasz-1979}  implies $\vartheta\leq \vartheta_{sp}$, proving that $\vartheta_{sp}=\vartheta$.
This conclusively shows that the sphere-packing bound, when applied to classical-quantum channels, gives precisely Lov\'asz' bound to $C_0$ and that pure state channels suffice for this purpose.
This also implies that for Lov\'asz's optimal representations there is always a rank-one minimizing $F$  in \eqref{eq:minmaxQRinfty}. It is worth pointing out that this is not true in general and that strict inequality holds in \eqref{eq:RinftyvsV} for some channels.


\bibliographystyle{IEEEtran}

\end{document}